\documentclass[a4paper,12pt,twoside]{article}
\title{\textbf{Decoding binary Reed-Muller codes via Groebner bases}}
\author{Harinaivo ANDRIATAHINY$^{(1)}$\\e-mail : hariandriatahiny@gmail.com\\Jean Jacques Ferdinand RANDRIAMIARAMPANAHY$^{(2)}$\\e-mail : randriamiferdinand@gmail.com\\Toussaint Joseph RABEHERIMANANA$^{(3)}$\\e-mail : rabeherimanana.toussaint@yahoo.fr\\\\$^{1,2,3}$Mention : Mathematics and Computer Science,\\Domain : Sciences and Technologies,\\University of Antananarivo, Madagascar}
\usepackage{geometry}
\usepackage{amsmath,amssymb}
\usepackage{amscd,amsfonts}
\usepackage{amsthm}
\usepackage[english]{babel}
\usepackage[T1]{fontenc}
\usepackage[ansinew]{inputenc}
\swapnumbers\theoremstyle{plain}
\swapnumbers
\theoremstyle{plain}

\usepackage[pdfborder={0 0 0}]{hyperref}
\usepackage{hyperref}

\newtheorem{thm}{Theorem}[section]
\newtheorem{prop}[thm]{Proposition}

\newtheorem{cor}[thm]{Corollary}

\DeclareMathOperator{\rad}{rad}
\DeclareMathOperator{\lcm}{lcm}

\DeclareMathOperator{\lt}{\mathrm{lt}}

\DeclareMathOperator{\lc}{\mathrm{lc}}
\DeclareMathOperator{\lm}{\mathrm{lm}}

\theoremstyle{definition}

\newtheorem{ex}[thm]{Example}
\newtheorem{rem}[thm]{Remark}

\DeclareMathOperator{\card}{card}

\begin{document}

\maketitle

\begin{abstract}
The binary Reed-Muller codes can be characterized as the radical powers of a modular algebra. We use the Groebner bases to decode these codes.
\end{abstract}

\section{Introduction}
S.D. Berman \cite{Berman} showed that the binary Reed-Muller codes may be described as the powers of the radical of the modular group algebra $\mathbb{F}_2[G]$, where $G$ is an elementary abelian 2-group. $\mathbb{F}_2[G]$ is isomorphic to the quotient ring $A=\mathbb{F}_2[X_1,\dots, X_m]\slash(X_1^2-1,\dots,X_m^2-1)$. The Jennings basis of $M^l$, where $M$ is the radical of $A$, is a linear basis of $M^l$ over $\mathbb{F}_2$. We use the fact that from the Jennings basis of $M^l$, one can construct a basis for the ideal $M^l$, and we utilize  the properties of the Groebner basis to establish a decoding algorithm for the binary Reed-Muller codes.  

\section{A division algorithm}
Let us start with some definitions about monomial orderings and division in a multivariate polynomial ring. Details can be found in \cite{clo}.\\
Let $k$ be an arbitrary field. A monomial in the $m$ variables $X_1,\dots,X_m$ is a product of the form $X_1^{\alpha_1}X_2^{\alpha_2}\dots X_m^{\alpha_m}$, where all the exponents $\alpha_1,\dots,\alpha_m$ are non negative integers. Let $\alpha = (\alpha_1,\dots,\alpha_m)$ be an m-tuple of non negative integers. We set $X^\alpha = X_1^{\alpha_1}X_2^{\alpha_2}\dots X_m^{\alpha_m}$. When $\alpha = (0,\dots,0)$, note that $X^\alpha = 1$. We also let $\mid\alpha \mid = \alpha_1+\dots+\alpha_m$ denote the total degree of the monomial $X^\alpha$.\\
A polynomial $f$ in $X_1,\dots,X_m$ with coefficients in $k$ is a finite linear combination with coefficients in $k$ of monomials. A polynomial $f$ will be written in the form $f = \sum_{\alpha}^{}a_\alpha X^\alpha$, $a_\alpha\in k$, where the sum is over a finite number of m-tuples $\alpha = (\alpha_1,\dots,\alpha_m)$.\\
$k[X_1,\dots,X_m]$ denotes the ring of all polynomials in $X_1,\dots,X_m$ with coefficients in $k$. A monomial ordering on $k[X_1,\dots,X_m]$ is any relation $>$ on $\mathbb{N}^m$, or equivalently, any relation on the set of monomials $X^\alpha$, $\alpha\in\mathbb{N}^m$, satisfying : 
\begin{itemize}
\item[(i)] $>$ is a total ordering on $\mathbb{N}^m$,
\item[(ii)] if $\alpha>\beta$ and $\gamma\in\mathbb{N}^m$, then $\alpha+\gamma>\beta+\gamma$,
\item[(iii)] $>$ is a well-ordering on $\mathbb{N}^m$.
\end{itemize}
A first example is the lexicographic order. Let $\alpha = (\alpha_1,\dots,\alpha_m)$ and $\beta = (\beta_1,\dots,\beta_m)\in\mathbb{N}^m$. We say $\alpha>_{lex}\beta$ if, in the vector difference $\alpha - \beta\in\mathbb{Z}^m$, the left-most nonzero entry is positive. We will write $X^\alpha>_{lex} X^\beta$ if $\alpha>_{lex}\beta$.\\
A second example is the graded lexicographic order. Let $\alpha, \beta\in \mathbb{N}^m$, we say $\alpha>_{grlex}\beta$ if $\mid\alpha\mid = \sum_{i=1}^m\alpha_i>\mid\beta\mid = \sum_{i=1}^m\beta_i$, or $\mid\alpha\mid=\mid\beta\mid$ and $\alpha>_{lex}\beta$.\\
Let $f=\sum_\alpha a_\alpha X^\alpha$ be a nonzero polynomial in $k[X_1,\dots,X_m]$ and let $>$ be a monomial order. The multidegree of $f$ is $multideg(f)= \max(\alpha\in\mathbb{N}^m/ a_\alpha\neq 0)$, the maximum is taken with respect to $>$. The leading coefficient of $f$ is $\lc(f)=a_{multideg(f)}\in k$. The leading monomial of $f$ is $\lm(f)=X^{multideg(f)}$.  The leading term of $f$ is $\lt(f)=\lc(f).\lm(f)$.
\begin{thm}
Fix a monomial order $>$ on $\mathbb{N}^m$, and let $F=(f_1,\dots,f_s)$ be an ordered s-tuple of polynomials in $k[X_1,\dots,X_m]$. Then every $f\in k[X_1,\dots,X_m]$ can be written as $f=a_1f_1+\dots+a_sf_s+r$, where $a_i, r\in k[X_1,\dots,X_m]$, and either $r=0$ or $r$ is a linear combination, with coefficients in $k$, of monomials, none of which is divisible by any of $\lt(f_1),\dots, \lt(f_s)$. We will call $r$ a remainder of $f$ on division by $F$. Furthermore, if $a_if_i\neq 0$, then we have $multideg(f)\geq multideg(a_if_i)$.
\end{thm} 
\begin{rem}\label{rem1}
The operation of computing remainders on division by $F=(f_1,\dots, f_s)$ is linear over $k$. That is, if the remainder on division of $g_i$ by $F$ is $r_i$, $i=1, 2$, then, for any $c_1, c_2\in k$, the remainder on division of $c_1g_1+c_2g_2$ is $c_1r_1+c_2r_2$.
\end{rem}
\section{Groebner bases}
In this section, we recall some basic properties of Groebner bases. Details and proofs can be found in \cite{clo}.\\
Let $I\subseteq k[X_1,\dots,X_m]$ be an ideal other than $\{0\}$. We denote by $\lt(I)$ the set of leading terms of elements of $I$. Thus, $\lt(I)=\{cX^\alpha/ \text{there exists}\  f\in I\ \  \text{with}\ \lt(f)=cX^\alpha\}$. We denote by $\langle \lt(I)\rangle$ the ideal of $k[X_1,\dots,X_m]$ generated by the elements of $\lt(I)$.
\begin{thm}[Hilbert Basis Theorem]
Every ideal $I\subseteq k[X_1,\dots,X_m]$ has a finite generating set. That is, $I=\langle g_1,\dots, g_t\rangle$ for some $g_1,\dots, g_t\in I$.
\end{thm}
Fix a monomial order. A finite subset $G=\{g_1,\dots,g_t\}$ of an ideal $I\subseteq k[X_1,\dots,X_m]$ is said to be a Groebner basis if $\langle\lt(g_1),\dots,\lt(g_t)\rangle=\langle\lt(I)\rangle$.
\begin{cor}
Fix a monomial order. Then every ideal $I\subseteq k[X_1,\dots,X_m]$ other than $\{0\}$ has a Groebner basis. Furthermore, any Groebner basis for an ideal $I$ is a basis of $I$.
\end{cor}
\begin{prop}\label{prop1}
Let $G=\{g_1,\dots,g_t\}$ be a Groebner basis for an ideal $I\subseteq k[X_1,\dots,X_m]$ and let $f\in k[X_1,\dots,X_m]$. Then there is a unique $r\in k[X_1,\dots,X_m]$ with the following properties :\\
$(i)$\ \ No term of $r$ is divisible by any of $\lt(g_1),\dots, \lt(g_t)$.\\
$(ii)$\ \ There is $g\in I$ such that $f=g+r$.\\
In particular, $r$ is the remainder on division of $f$ by $G$ no matter how the elements of $G$ are listed when using the division algorithm. 
\end{prop}
\begin{cor}\label{cor3}
Let $G=\{g_1,\dots,g_t\}$ be a Groebner basis for an ideal $I\subseteq k[X_1,\dots,X_m]$ and let $f\in k[X_1,\dots,X_m]$. Then $f\in I$ if and only if the remainder on division of $f$ by $G$ is zero.
\end{cor}
We will write $\overline{f}^F$ for the remainder on division of $f$ by the ordered s-tuple $F=(f_1,\dots, f_s)$. If $F$ is a Groebner basis for $\langle f_1,\dots, f_s\rangle$, then we can regard $F$ as a set without any particular order.\\ 
Let $f, g\in k[X_1,\dots,X_m] $ be nonzero polynomials. If $multideg(f)=\alpha$ and $multideg(g)=\beta$, then let $\gamma=(\gamma_1,\dots, \gamma_m)$ where $\gamma_i=\max(\alpha_i, \beta_i)$ for each $i$. We call $X^\gamma$ the least common multiple of $\lm(f)$ and $\lm(g)$, written $X^\gamma=\lcm(\lm(f),\lm(g))$. The S-polynomial of $f$ and $g$ is the combination 
\[S(f,g)=\dfrac{X^\gamma}{\lt(f)}.f-\dfrac{X^\gamma}{\lt(g)}.g\]
\begin{thm}
Let $I$ be a polynomial ideal. Then a basis $G=\{g_1,\dots,g_t\}$ for $I$ is a Groebner basis for $I$ if and only if for all pairs $i\neq j$, the remainder on division of $S(g_i,g_j)$ by $G$ listed in some order is zero.
\end{thm}
 \begin{rem}
 Let $I\subseteq k[X_1,\dots,X_m]$ be an ideal, and let $G$ be a Groebner basis of $I$. Then\\
 $(i)$\ \ $\overline{f}^G=\overline{g}^G$ if and only if $f-g\in I$\\
 $(ii)$\ \ $\overline{f+g}^G= \overline{f}^G+\overline{g}^G$\\
 $(iii)$\ \ $\overline{f.g}^G=\overline{\overline{f}^G.\overline{g}^G}^G$
  \end{rem}
  A reduced Groebner basis for a polynomial ideal $I$ is a Groebner basis $G$ for $I$ such that :\\
  $(i)$\ \ $\lc(p)=1$ for all $p\in G$\\
  $(ii)$\ \ For all $p\in G$, no monomial of $p$ lies in $\langle \lt(G-\{p\})\rangle$
  \begin{prop}
  Let $I\neq \{0\}$ be a polynomial ideal. Then, for a given nomomial ordering, $I$  has a unique reduced Groebner basis.
  \end{prop}

\section{Binary Reed-Muller codes}
In this section, we recall some basic properties of the Reed-Muller codes.\\
Consider the ideal $\mathcal{I}=\langle X_1^2-1,\dots,X_m^2-1\rangle$ of the polynomial ring $\mathbb{F}_2[X_1,\dots,X_m]$. We use the quotient ring $\mathcal{A}=\mathbb{F}_2[X_1,\dots,X_m]/\mathcal{I}$ as the ambiant space for the codes of length $2^m$ over $\mathbb{F}_2$. We set $x_1=X_1+\mathcal{I},\dots,x_m=X_m+\mathcal{I}$ and we obtains $x_1^2=1,\dots, x_m^2=1$. Then, we have \[\mathcal{A}=\left\{a(x)=\sum_{i_1=0}^{1}\dots\sum_{i_m=0}^{1}a_{i_1\dots i_m}x_1^{i_1}\dots x_m^{i_m}/a_{i_1\dots i_m}\in\mathbb{F}_2\right\}.\] We write also $a(x)=\displaystyle{\sum_{i\in\Gamma}a_ix^i}$ where $i=(i_1,\dots, i_m)\in\Gamma\subseteq(\{0, 1\})^m$ and $x^i=x_1^{i_1}\dots x_m^{i_m}$. 
Let us order the monomials in the set $\{x_1^{i_1}\dots x_m^{i_m}/0\leq i_1,\dots,i_m\leq 1\}$ with the graded lexicographic order (grlex) such that $x_1>_{lex}x_2>\dots>_{lex}x_m$. We have the following correspondance : 
\begin{equation}
\mathcal{A}\ni a(x)=\sum_{i_1=0}^{1}\dots\sum_{i_m=0}^{1}a_{i_1\dots i_m}x_1^{i_1}\dots x_m^{i_m}\longleftrightarrow a=(a_{i_1\dots i_m})_{0\leq i_1,\dots,i_m\leq 1}\in(\mathbb{F}_2)^{2^m}
\end{equation}
\begin{thm}\label{thm4.1}
$\mathcal{A}$ is a local ring with maximal ideal $M=\rad(\mathcal{A})$ the radical of $\mathcal{A}$. For each integer $l$ such that $0\leq l\leq m$, a linear basis of the radical power $M^l$ of $M$ over $\mathbb{F}_2$ is given by
\begin{equation}
B_l:=\left\{(x_1-1)^{i_1}\dots (x_m-1)^{i_m}/0\leq i_1,\dots,i_m\leq 1,i_1+\dots+i_m\geq l \right\}
\end{equation}
$B_l$ is called the Jennings basis of $M^l$, and we have the sequence of ideals
\begin{equation}
\{0\}=M^{m+1}\subset M^m\subset\dots\subset M^2\subset M\subset \mathcal{A}
\end{equation}
 \end{thm}
 \begin{cor}
 We have $\dim_{\mathbb{F}_2}(M^l)=\dbinom{m}{l}+\dbinom{m}{l+1}+\dots+\dbinom{m}{m}$.
 \end{cor}
\noindent Let $P(m,2)$ be the set of all reduced form polynomials in $m$ variables $Y_1,\dots, Y_m$ over $\mathbb{F}_2$ :
 \begin{center}
 $P(m,2) := \left\{P(Y_1,\dots,Y_m)=\sum_{i_1=0}^{1}\dots\sum_{i_m=0}^{1}u_{i_1\dots i_m}Y_1^{i_1}\dots Y_m^{i_m}/u_{i_1\dots i_m}\in\mathbb{F}_2\right\}$
 \end{center}
 We set $\beta_0=0$ and $\beta_1=1$.\\
 As vector spaces over $\mathbb{F}_2$, we have the following isomorphism :
 \begin{align*}
 \phi\; :\;\;& P(m,2)\hspace{1.1cm} \longrightarrow \hspace{2cm} \mathcal{A}  \\
      & P(Y_1,\dots, Y_m) \longmapsto   \displaystyle{\sum_{i_1=0}^{1}\dots\sum_{i_m=0}^{1}P(\beta_{i_1},\dots,\beta_{i_m})x_1^{i_1}\dots x_m^{i_m}}
 \end{align*}
 Let $\nu$ be an integer such that $0\leq \nu\leq m$. Denote by $P_\nu(m,2)$ the subspace of $P(m,2)$ generated by the monomials of total degree $\nu$ or less. The $\nu$th-order Reed-Muller code of lenght $2^m$ over $\mathbb{F}_2$ is defined by 
 \begin{center}
 $\mathcal{C}_\nu(m,2) := \left\{(P(\beta_{i_1},\dots,\beta_{i_m}))_{0\leq i_1,\dots,i_m\leq 1}/P(Y_1,\dots,Y_m)\in P_\nu(m, 2)\right\}$
 \end{center}
 $\mathcal{C}_\nu(m,2)$ is a subspace of $(\mathbb{F}_2)^{2^m}$, and we have the following sequence
 \begin{equation}
  \{0\}\subset \mathcal{C}_0(m,2)\subset\mathcal{C}_1(m,2)\dots\subset\mathcal{C}_{m-1}(m,2)\subset\mathcal{C}_m(m,2)=(\mathbb{F}_2)^{2^m}
  \end{equation}
 \begin{thm}[Berman] \ \ \\
We have $M^l=\mathcal{C}_{m-l}(m,2)$, for $l$ such that\ \ $0\leq l\leq m$.
 \end{thm}
\noindent The weight of a word $v=(v_1,v_2,\dots,v_{2^m})\in(\mathbb{F})^{2^m}$ is defined by $\omega(v):=\card(\{i/v_i\neq 0\})$
\begin{thm}
 The minimal weight of the Reed-Muller code $M^l$ is \[d=2^l,\ \ 0\leq l\leq m.\]
 \end{thm}
\noindent $M^l$ is a $t$-error correcting code where $t$ is the maximal integer such that $2t+1\leq 2^l.$

\section{Decoding}
We now present our main results and the decoding algorithm.\\
 From now on, we set 
 \begin{equation}
 G_l:=\{(x_1-1)^{i_1}\dots(x_m-1)^{i_m}/ 0\leq i_1,\dots,i_m\leq 1, i_1+\dots+i_m=l\}
 \end{equation}
 \begin{prop}\ \ \\
 $G_l$ is a basis for the ideal $M^l$ $(0\leq l \leq m)$.
 \end{prop}
\begin{proof}
 Since $G_l\subseteq M^l$ and $M^l$ is an ideal of $\mathcal{A}$, then $\mathcal{A}.g\subseteq M^l$ for all $g\in G_l$. Thus $\displaystyle{\sum_{g\in G_l}{}\mathcal{A}.g\subseteq M^l}$.\\
 Conversely, since $B_l$ is a linear basis of $M^l$ over $\mathbb{F}_2$, then each element of $M^l$ is a sum of elements in $B_l$. Every element of $B_l$ can be written as a product $a.g$, where $a\in \mathcal{A}$ and $g\in G_l$. Thus, we have $M^l\subseteq\displaystyle{\sum_{g\in G_l}{}\mathcal{A}.g}$.
\end{proof}
 Let us fix an integer $l$ such that $1\leq l\leq m$.
 We set
 \begin{equation}
  G=\big\{(X_1-1)^{i_1}\dots (X_m-1)^{i_m}\slash \  0\leq i_1,\dots,i_m\leq 1, i_1+\dots+i_m=l\big\}\subseteq\mathbb{F}_2[X_1,\dots,X_m]
  \end{equation}
   and
   \begin{equation}
   H=\{X_1^2-1,\dots, X_m^2-1\}\subseteq\mathbb{F}_2[X_1,\dots,X_m].
   \end{equation}
 Let $f(X)\in\mathbb{F}_2[X_1,\dots,X_m]$, then $f(X)=\displaystyle{\sum_{\alpha\in\Lambda}X^\alpha}$ where $\Lambda\in\mathbb{N}^m$, \ $X^\alpha=X_1^{\alpha_1}\dots X_m^{\alpha_m}$ and $\card(\Gamma)<+\infty$.\\
  Let $E=\{1,2,\dots, m\}$. For each subset $I\subseteq E$, we define $X_I=\displaystyle{\prod_{i\in I}X_i}$ with $X_\emptyset=1$, and \\ $g_I=\displaystyle{\prod_{i\in I}(X_i-1)}$ with $g_\emptyset=1$. The elements of $G$ are ordered as follows :\[g_I > g_J\iff X_I > X_J.\]
 \begin{prop}\ \ \\
 $G$ is a reduced Groebner basis for the ideal $\langle G\rangle\subseteq\mathbb{F}_2[X_1,\dots,X_m]$ \ \ $(0\leq l\leq m)$.
 \end{prop}
\begin{proof} Let $I, J\subseteq E$ such that $\card( I )=\card( J)= l$. We have
 \begin{align*}
 S(g_I,g_J) & =\dfrac{\lcm(\lm(g_I),\lm(g_J))}{\lt(g_I)}.g_I-\dfrac{\lcm(\lm(g_I),\lm(g_J))}{\lt(g_J)}.g_I\\
 & =X_{(I\cup J)\backslash I}.g_I-X_{(I\cup J)\backslash J}.g_J\\
 &= \prod_{i\in(I\cup J)\backslash I}[(X_i-1)+1].g_I-\prod_{i\in(I\cup J)\backslash J}[(X_i-1)+1].g_J\\
 & = (\sum_{K\in\mathcal{P}((I\cup J)\backslash I)}g_K)g_I-(\sum_{K\in\mathcal{P}((I\cup J)\backslash J)}g_K)g_J\\
 & = \sum_{K\in\mathcal{P}((I\cup J)\backslash I)}g_K.g_I-\sum_{K\in\mathcal{P}((I\cup J)\backslash J)}g_K.g_J
 \end{align*}
 It is clear that $\overline{g_K.g_I}^G=0$ for all 
 $K\in\mathcal{P}((I\cup J)\backslash I)$ and $\overline{g_K.g_J}^G=0$ for all 
  $K\in\mathcal{P}((I\cup J)\backslash J)$ and by remark \ref{rem1}, we have $\overline{S(g_I,g_J)}^G=0$.
\end{proof}
  For each subset $I\subseteq E$, we define $\widehat{I}\subseteq \mathcal{P}(I)$ by $\overline{X_I}^G=\displaystyle{\sum_{L\in\widehat{I}}}X_L$.
  \begin{rem}
  For $f(X)\in\mathbb{F}_2[X_1,\dots,X_m]$, if $r(X)=\overline{f}^H$, then $r(X)=\displaystyle{\sum_{\alpha\in\Lambda}X^\alpha}$, $\Lambda\subseteq(\{0, 1\})^m$, $X^\alpha=X_1^{\alpha_1}\dots X_m^{\alpha_m}$ and $f(X)+\mathcal{I}=r(X)+\mathcal{I}$. Since $x_1=X_1+\mathcal{I},\dots,x_m=X_m+\mathcal{I}$, then $r(X)+\mathcal{I}=r(x)=\displaystyle{\sum_{\alpha\in\Lambda}x^\alpha}$ where $x^\alpha=x_1^{\alpha_1}\dots x_m^{\alpha_m}$.
   \end{rem}
   \begin{rem}\label{rem5}
   If $f(X)=\displaystyle{\sum_{\alpha\in\Gamma}X^\alpha}\in\mathbb{F}_2[X_1,\dots,X_m]$ with $\Gamma\subseteq(\{0, 1\})^m$, then $\overline{f(X)}^H=f(X)$ and $\overline{f(X)}^{G\cup H}=\overline{f(X)}^G$. Thus, in this case $f(X)\in\langle G\rangle\iff f(X)\in\langle G\cup H\rangle$. 
   
   \end{rem}
   
   Consider the ideal $\langle G\cup H\rangle\slash\mathcal{I}$ of $\mathcal{A}=\mathbb{F}_2[X_1,\dots,X_m]\slash\mathcal{I}$, we have $M^l=\langle G\cup H\rangle\slash\mathcal{I}$.\\ For each $f(x)=\displaystyle{\sum_{\alpha\in\Gamma}a_\alpha x^\alpha\in \mathcal{A}}$ where $\Gamma\subseteq(\{0, 1\})^m$, we always take $f(X)=\displaystyle{\sum_{\alpha\in\Gamma}a_\alpha X^\alpha}$ $\in$ $\mathbb{F}_2[X_1,\dots,X_m]$ as its representative modulo $\mathcal{I}=\langle  H\rangle$, and we denote by $\overline{f}^G$ the remainder on division of $f(X)$ by $G$.
   \begin{prop}
   Let $f(x)\in \mathcal{A}$. Then, \[f(x)\in M^l\ \ \text{if and only if}\ \ \overline{f}^G=0.\]
    \end{prop}
\begin{proof} We have $f(x)\in M^l$ if and only if $ f(X)\in\langle G\cup H\rangle$. Thus, the Proposition follows from Remark \ref{rem5} and Corollary \ref{cor3}.
\end{proof}
\begin{prop}
    The subset $H=\{X_1^2-1,\dots,X_m^2-1\}\subseteq\mathbb{F}_2[X_1,\dots,X_m]$ is a Groebner basis for the ideal $\langle H \rangle$.
\end{prop}
\begin{proof}
Let  $i$ and $j\in\{1,\dots,m\}$ with $i<j$, we have
\begin{align*} S(X_i^2-1,X_j^2-1) &=\dfrac{X_i^2X_j^2}{X_i^2}(X_i^2-1)-\dfrac{X_i^2X_j^2}{X_j^2}(X_j^2-1)\\ &=X_j^2(X_i^2-1)-X_i^2(X_j^2-1)\\ &=(X_j^2-1+1)(X_i^2-1)-(X_i^2-1+1)(X_j^2-1)\\
    & =(X_j^2-1)(X_i^2-1)+(X_i^2-1)- (X_i^2-1)(X_j^2-1)+(X_j^2-1)\\
    & = (X_i^2-1)+(X_j^2-1)\in\langle H \rangle
\end{align*}
    Thus, $\overline{S(X_i^2-1,X_j^2-1)}^H=0$.
\end{proof}
    \begin{prop}
    $G\cup H$ is a Groebner basis for the ideal $\langle G\cup H\rangle\subseteq\mathbb{F}_2[X_1,\dots,X_m]$.
    \end{prop}
\begin{proof}  Let $I\subseteq E=\{1,\dots, m\}$ with $\card(I)=l$. It's enough to prove that the remainder on division of $S(g_I,X_j^2-1)$ by $G\cup H$ is zero, where $g_I\in G$ and $X_j^2-1\in H$ for all $j=1,\dots, m$.\\
    For $j\in I$, we have
    \begin{align*} S(g_I,X_j^2-1) &=\dfrac{X_IX_{\{j\}}}{X_I}.g_I-\dfrac{X_IX_{\{j\}}}{X_j^2}(X_j^2-1)\\ &=X_j.g_I-X_{I\backslash\{j\}}(X_j^2-1)\\
     &=(X_j-1+1)g_I-(X_j^2-1)\displaystyle{\prod_{k\in I\backslash\{j\}}(X_k-1+1)}\\
        & =(X_j-1)g_I+g_I-[(X_j^2-1)(X_{k_1}-1)+\dots+(X_j^2-1)(X_{k_{l-1}}-1) \\ & \  +(X_j^2-1)(X_{k_1}-1)(X_{k_2}-1)+\dots+(X_j^2-1)\displaystyle{\prod_{k\in I\backslash\{j\}}(X_k-1)}+(X_j^2-1)]\\
        & \text{where}\  \{k_1,k_2,\dots,k_{l-1}\}= I\backslash\{j\} 
        \end{align*}
  For $j\in E\backslash I$, we have
        \begin{align*} S(g_I,X_j^2-1) &=\dfrac{X_IX_{\{j\}}X_{\{j\}}}{X_I}.g_I-\dfrac{X_IX_{\{j\}}X_{\{j\}}}{X_{\{j\}}X_{\{j\}}}(X_j^2-1)\\
         &=X_j^2.g_I-X_{I}(X_j^2-1)\\
         &=(X_j^2-1+1)g_I-(X_j^2-1)\displaystyle{\prod_{k\in I}(X_k-1+1)}\\
            & =(X_j^2-1)g_I+g_I-[(X_j^2-1)(X_{p_1}-1)+\dots+(X_j^2-1)(X_{p_{l}}-1) \\ & \  +(X_j^2-1)(X_{p_1}-1)(X_{p_2}-1)+\dots+(X_j^2-1)\displaystyle{\prod_{k\in I}(X_k-1)}+(X_j^2-1)]\\
            & \text{where}\  \{p_1,p_2,\dots,p_{l}\}= I. 
            \end{align*}
  Then for all $j$ and $I$, we have $\overline{S(g_I,X_j^2-1)}^{G\cup H}=0$.
\end{proof}
  \begin{prop}\label{prop5}
 Let $I, J\subseteq E=\{1,\dots,m\}$, $I\neq J$. We have $\overline{X_I+X_J}^G=\displaystyle{\sum_{L\in\widehat{I}\Delta\widehat{J}}}X_L$ in $\mathbb{F}_2[X]$ where $\widehat{I}\Delta\widehat{J}=(\widehat{I}\backslash\widehat{J})\cup(\widehat{J}\backslash\widehat{I})$.
  \end{prop}
\begin{proof}
  By Remark \ref{rem1}, we have $\overline{X_I+X_J}^G=\overline{X_I}^G+\overline{X_J}^G=\displaystyle{\sum_{L\in\widehat{I}}}X_L+\displaystyle{\sum_{L\in\widehat{J}}}X_L=\displaystyle{\sum_{L\in(\widehat{I}\backslash\widehat{J})\cup(\widehat{I}\cap\widehat{J})}}X_L+\displaystyle{\sum_{L\in(\widehat{J}\backslash\widehat{I})\cup(\widehat{I}\cap\widehat{J})}}X_L= \displaystyle{\sum_{L\in(\widehat{I}\backslash\widehat{J})}}X_L+\displaystyle{\sum_{L\in(\widehat{J}\backslash\widehat{I})}}X_L+2\displaystyle{\sum_{L\in(\widehat{I}\cap\widehat{J})}}X_L=\displaystyle{\sum_{L\in(\widehat{I}\backslash\widehat{J})\cup(\widehat{J}\backslash\widehat{I})}}X_L=\displaystyle{\sum_{L\in(\widehat{I}\Delta\widehat{J})}}X_L$.
\end{proof}
\begin{prop}
 Let $I\subseteq E$, \\
$-$ if $\card(I)<l$, then $\widehat{I}=\{I\}$.\\
 $-$ if $\card(I)=l$, then $\widehat{I}=\mathcal{P}(I)-\{I\}$ and $\omega(\overline{X_I}^G)>t$.
 
\end{prop}
\begin{proof}
Suppose that $\card(I)<l$. Since $X_I$ is not divisible by any of $\lt(g)$,   $g\in G$, then $\overline{X_I}^G=X_I$.\\
 In the second case, we have $g_I= \displaystyle{\sum_{L\in\mathcal{P}(I)}}X_L $. Then we can write $X_I= g_I+ \displaystyle{\sum_{L\in\mathcal{P}(I)\backslash\{I\}}}X_L$. It is clear that no term of $\displaystyle{\sum_{L\in\mathcal{P}(I)\backslash\{I\}}}X_L$ is divisible by any of $\lt(g)$,   $g\in G$. Furthermore, we have $g_I\in \langle G\rangle$. Thus, by Proposition \ref{prop1},       $\overline{X_I}^G=\displaystyle{\sum_{L\in\mathcal{P}(I)\backslash\{I\}}}X_L$. We have $\omega(\displaystyle{\sum_{L\in\mathcal{P}(I)\backslash\{I\}}}X_L)=\card(\mathcal{P}(I)\backslash\{I\})=2^{\card(I)}-1=2^l-1$. Since $t$ is the maximal integer such that $2t+1\leq 2^l$, i.e $2t\leq 2^l-1$, thus $t<2^l-1$. Then $\omega(\overline{X_I}^G)>t$.
\end{proof}
  
\begin{prop}\label{prop8}
 Let $c(x)\in M^l$ be a transmitted codeword and $v(x)\in \mathcal{A}$ the received vector. Since our code is t-errors correcting, we write $v(x)=c(x)+e(x)$ with $\omega(e)\leq t$. Thus, we have $\overline{v}^G=\overline{e}^G$. And \\
 $-$ \ if $\overline{v}^G=0$, then $c=v$\\
 $-$ \ if $\overline{v}^G=\displaystyle{\sum_{L\in(\dots((\widehat{I_1}\Delta\widehat{I_2})\Delta\widehat{I_3})\Delta\dots\Delta\widehat{I_k})}}x_L$ with $k\leq t$ and $I_1,I_2,\dots, I_k\subseteq E$ are pairwise distincts, then $e(x)=x_{I_1}+x_{I_2}+\dots+x_{I_k}$ i.e $v(x)$ contains $k$ errors located at $x_{I_1},x_{I_2},\dots,x_{I_k}$.
  \end{prop}
\begin{proof}
  Since $c\in M^l$ and $G$ is a Groebner basis of $\langle G \rangle$ then $\overline{c}^G=0$. It follows that $\overline{v}^G=\overline{c+e}^G= \overline{c}^G+\overline{e}^G=\overline{e}^G$.\\
  $-$ If $\overline{v}^G=0$, then $\overline{e}^G=0$. Thus $e(x)\in M^l$. And since $\omega(e)\leq t<2^l=d_{\min}(M^l)$, then $e=0$. So we have $c=v$.\\
  $-$ If $\overline{v}^G=\displaystyle{\sum_{L\in(\dots((\widehat{I_1}\Delta\widehat{I_2})\Delta\widehat{I_3})\Delta\dots\Delta\widehat{I_k})}}X_L$, by Proposition \ref{prop5}, we  have $\overline{v}^G=\overline{X_{I_1}+X_{I_2}+\dots+X_{I_k}}^G$. Thus, $v(X)-(X_{I_1}+X_{I_2}+\dots+X_{I_k})=c^\prime(X)\in \langle G\rangle$, i.e  $\overline{c^\prime(X)}^G=0$, then  $c^\prime(x)\in M^l$. We have $v(x)=c^\prime(x)+e^\prime(x)$ with $e^\prime=x_{I_1}+x_{I_2}+\dots+x_{I_k}$.\\ Since $v(x)=c(x)+e(x)$, then $c(x)+e(x)=c^\prime(x)+e^\prime(x)$, this implies that $e^\prime(x)-e(x)=c(x)-c^\prime(x)\in M^l$. On the other hand, we have $\omega(e^\prime(x)-e(x))\leq \omega(e^\prime(x))+\omega(e(x))\leq k+t\leq 2t<2^l=d_{\min}(M^l)$. Thus, $e^\prime(x)-e(x)=0$, i.e $e^\prime(x)=e(x)$.
\end{proof}
\begin{prop}
 Let $v(x)=c(x)+e(x)$ with $c(x)\in M^l$ and $\omega(e)\leq t$ , if $\omega(\overline{v}^G)=k\leq t$, then $e(X)=\overline{v(X)}^G$, and $e=x_{I_1}+x_{I_2}+\dots+x_{I_k}$ such that $\forall \tau=1,\dots, k$, $\card(I_\tau)<l$.\\
  If $\omega(\overline{v(X)}^G)>t$, then there exists at least one index $\tau\in\{1,\dots,k\}$ with $k\leq t$ such that $\card(I_\tau)\geq l$ and $e=\displaystyle{\sum_{j=1}^{k}x_{I_j}}$ where $I_1,\dots, I_k\subseteq E$.
\end{prop}
\begin{proof}
  We have $v(X)=c^\prime(X)+\overline{v(X)}^G$, with $c^\prime(x)\in M^l$, and $v=c+e$. So, $c+e=c^\prime+\overline{v}^G$. And thus, $\overline{v}^G-e=c-c^\prime\in M^l$. Since $\omega(\overline{v}^G)\leq t$ and $\omega(e)\leq t$, then $\omega(\overline{v}^G-e)\leq \omega(\overline{v}^G)+\omega(e)\leq 2t<2t+1\leq 2^l=d_{\min}(M^l)$. Thus, $\overline{v}^G-e=0$, i.e $e=\overline{v}^G$. Since no term of $\overline{v}^G$ is divisible by any of $\lt(g)$, $g\in G$, then we have $\card(I_\tau)<l$ for all $\tau=1,\dots, k$.\\
 For the second case, since $\overline{v}^G=\overline{e}^G$, then $\omega(\overline{e}^G)=\omega(\overline{v}^G)>t$.
    And since $\omega(e)\leq t$, then $e=\displaystyle{\sum_{j=1}^{k}X_{I_j}}$ with $k\leq t$ and $I_1,\dots, I_k\subseteq E$. If $\card(I_j)<l $ for all $j=1,\dots, k$, then $\overline{X_{I_j}}^G=X_{I_j}$ \ ($1\leq j\leq k$), and  $\overline{e}^G=\overline{\displaystyle{\sum_{j=1}^{k}X_{I_j}}}^G=\displaystyle{\sum_{j=1}^{k}\overline{X_{I_j}}^G} =\displaystyle{\sum_{j=1}^{k}X_{I_j}}=e$. Then, we have $\omega(\overline{e}^G)=\omega(e)=k\leq t$, a contradiction.
\end{proof}
  
\begin{prop}
 Let $v(x)=c(x)+e(x)$ with $c(x)\in M^l$ and $\omega(e)\leq t$, if $e=x_{i_1}+\dots+x_{i_k}$ with $k\leq t$ and there exists $\tau\in\{1,\dots, k\}$ such that $\card(I_\tau)\geq l$, then $\overline{v}^G>t$.
 \end{prop}
\begin{proof}
 We have $v=c+e$. Then, $\overline{v}^G=\overline{e}^G$. On the other hand, we have $v=c^\prime+\overline{v}^G$ with $c^\prime\in M^l$. Thus, $c+e=c^\prime+\overline{v}^G$. So $\overline{v}^G-e=c-c^\prime\in M^l$. If $\omega(\overline{v}^G)\leq t$, then $\omega(\overline{v}^G-e)\leq\omega(\overline{v}^G)+\omega(\overline{e}^G)\leq 2t<2t+1\leq 2^l=d_{\min}(M^l)$. Thus, $\overline{v}^G-e=0$,  i.e $\overline{e}^G=\overline{v}^G=e$, a contradiction because no term of $\overline{e}^G$ is divisible by any of $\lt(g)$, $g\in G$.
\end{proof}
\begin{cor}
Let $I\subseteq E$ such that $\card(I)>l$, then $\omega(\overline{X_I}^G)>t$.
\end{cor}
\begin{proof} If $e=X_I$, then $\omega(\overline{X_I}^G)=\omega(\overline{e}^G)=\omega(\overline{v}^G)>t$.
\end{proof}
\begin{thm}
 Let $v\in \mathcal{A}$ be a received vector wich contains at most $t$ errors where $t$ is the maximal integer such that $2t+1\leq 2^l$. Then $v$ can be decoded by the following algorithm :\\
 Input : \\
 $-$\  $v$\\
 $-$\  $G$, a reduced Groebner basis for $\langle G \rangle$\\
 $-$\  $\Omega=\{S\subseteq\mathcal{P}(\{1,\dots, m\})\slash\card(I)\geq l\ \  \text{for all}\ \  I\in S\}$\\
 Output : a codeword $c$\\
 BEGIN 
 
$-$ \  Compute $\overline{v}^G$

 $-$\  If $\omega(\overline{v}^G)\leq t$, then $c=v+\overline{v}^G$
 
 $-$\  Otherwise, find the element $S\in\Omega$ such that $\omega(\overline{v}^G-\overline{\displaystyle{\sum_{I\in S}X_I}}^G)\leq t -\card(S)$, 
 
 \hspace{0.5cm} then $c=v+\displaystyle{\sum_{I\in S}X_I}+\overline{v}^G-\displaystyle{\overline{\sum_{I\in S}X_I}}^G$.\\
 END
  \end{thm}
 In the case of the Reed-Muller code $M^2$ which is a one error correcting code, we have a simple decoding algorithm. 
 \begin{cor}
 Consider the Reed-Muller code $M^2$. Let  $v\in(\mathbb{F}_2)^{2^m}$ be a received vector which contains at most one error. Denote $v(x)$ the polynomial in $\mathcal{A}$ corresponding to $v$. We have $v(x)=c(x)+e(x)$ with $c(x)\in M^2$ and $\omega(e)\leq 1$.\\
 $-$ \ If $\overline{v}^G=0$ then $v=c$.\\
 $-$ \ If $\overline{v}^G=\displaystyle{\sum_{i\in I}X_{i}}$ or $\overline{v}^G=\displaystyle{\sum_{i\in I}X_{i}}+1$ with $I\subseteq E$, then $e=\displaystyle{\prod_{i\in I}x_i}$.
\end{cor}
 \begin{ex}
Consider the Groebner basis $G=\{X_1X_2+X_1+X_2+1,\ X_1X_3+X_1+X_3+1,\ X_2X_3+X_2+X_3+1\}$ for $\langle G \rangle\subseteq\mathbb{F}_2[X_1,X_2,X_3]$. Let $v=(1,0,1,0,0,0,1,0)$ is a received vector. Since $v(x)=x_1x_2x_3+x_1x_3+x_3$, then $\overline{v}^G=X_2+X_3+1$. Thus $e=x_2x_3$ and we obtain the codeword $c=(1,0,1,0,0,0,1,0)+(0,0,0,0,1,0,0,0)=(1,0,1,0,1,0,1,0)\in M^2=\mathcal{C}_{3-2}(3,2)$.
 \end{ex}

\end{document}